\documentclass[11pt]{article}

\usepackage[margin=1in]{geometry}
\usepackage{setspace}
\usepackage{microtype}

\usepackage{amsmath}
\usepackage{amssymb}
\usepackage{amsthm}
\usepackage{mathtools}
\usepackage{bm}

\usepackage{graphicx}
\usepackage{float}
\usepackage{booktabs}
\usepackage{array}
\usepackage{caption}
\usepackage{subcaption}
\usepackage{threeparttable}

\usepackage{tikz}
\usetikzlibrary{positioning,arrows.meta}

\usepackage{pgfplots}
\pgfplotsset{compat=1.18}

\usepackage{enumitem}

\usepackage[numbers]{natbib}
\usepackage{hyperref}

\hypersetup{
    colorlinks=true,
    linkcolor=blue,
    citecolor=blue,
    urlcolor=blue
}

\usepackage[T1]{fontenc}
\usepackage{lmodern}

\usepackage{authblk}

\newtheorem{proposition}{Proposition}
\newtheorem{lemma}{Lemma}

\newtheorem{remark}{Remark}
\newtheorem{definition}{Definition}

\title{\textbf{Changing the Game: Status-Quo Inertia, Institutional Design, and Equilibrium Transition}}

\title{Changing the Game: Status-Quo Inertia, Institutional Design, and Equilibrium Transition}

\author[1]{Madjid Eshaghi Gordji}
\author[2]{Esmaiel Abounoori}
\author[1,*]{Mohamadali Berahman}

\affil[1]{Faculty of Mathematics, Statistics and Computer Science, Semnan University, Semnan, Iran}
\affil[2]{Faculty of Economics, Semnan University, Semnan, Iran}

\affil[*]{Corresponding author: \texttt{mohamadali\_berahman@semnan.ac.ir}}

\usetikzlibrary{positioning}
\begin{document}

\maketitle


\begin{abstract}
Many economic interventions are designed as marginal changes in incentives. Yet in environments shaped by coordination, institutional persistence, and path dependence, such reforms often leave behavior largely unchanged. This paper studies interventions in games when equilibrium selection displays status-quo inertia: if the pre-intervention equilibrium remains a Nash equilibrium after policy, it continues to be selected. In such environments, price-based interventions and simple option expansion may fail even when they improve welfare in a partial-equilibrium sense. By contrast, interventions that modify the feasible action space---especially deletion and replacement interventions---can be substantially more effective because they remove the strategic basis for persistence. We develop a simple framework, derive general results, and provide complete proofs. We then illustrate the mechanism through stylized game-theoretic examples and selected applications from climate negotiations, disruptive innovation, military conflict, platform competition, financial reform, and technological rivalry. The examples clarify why addition-only reforms may leave the inherited equilibrium intact, while deletion and replacement can induce transition by eliminating or redesigning the actions that support the status quo. The analysis highlights a basic policy lesson: when inefficient equilibria are institutionally entrenched, the central problem is often not how to price the existing game more finely, but how to change the game itself.
\end{abstract}

\vspace{1em}

\noindent
\textbf{JEL Codes:}
C72, D02, D62, L51, Q58

\vspace{0.5em}

\noindent
\textbf{Keywords:}
equilibrium selection, status-quo inertia, institutional design, mechanism design, coordination failure, regulation, path dependence


\section{Introduction}

Since the foundational work of \citet{nash1950,nash1951}, game theory has provided economists with a language for understanding strategic interdependence. But from the beginning, one difficulty was apparent: equilibrium existence does not by itself identify the outcome an economy will reach. This concern led to a vast literature on equilibrium selection, focal points, and strategic refinement, from \citet{schelling1960} to \citet{harsanyi1988}. In parallel, mechanism design and implementation theory---shaped by the work of \citet{hurwicz1972}, \citet{myerson1982}, and \citet{maskin1999}---showed that outcomes depend not only on incentives, but also on the institutional structure of interaction.

This paper builds on those traditions while focusing on a practical issue that is central in modern policy environments: economies rarely choose among equilibria from a blank slate. They inherit an existing equilibrium. In many settings, that inherited outcome is supported by routines, legal continuity, sunk investments, complementary infrastructures, network effects, and shared expectations. As a consequence, even when a policy creates a better equilibrium, the economy may remain where it is.
We study interventions in games under \emph{status-quo inertia}. The idea is simple. If the pre-intervention equilibrium remains a Nash equilibrium after policy, it continues to be selected. This assumption is intentionally reduced-form, but it captures a broad class of phenomena that economists routinely observe: institutional stickiness, path dependence, lock-in, and equilibrium persistence.

The distinction matters because many standard interventions are designed as local changes in payoffs. Taxes, subsidies, fines, bonuses, and relative-price corrections all aim to improve incentives while leaving the set of feasible actions intact. In environments without serious coordination frictions, that may be enough. But when agents must move together, price corrections can be too weak to dislodge a prevailing equilibrium. A superior alternative may exist and still remain unrealized.

The paper's main claim is that in such environments, the decisive question is often not whether policy improves incentives, but whether policy changes the game deeply enough to eliminate the inherited equilibrium. This shifts attention from marginal payoff manipulation to institutional redesign. We compare four classes of interventions: price-only interventions, addition interventions, deletion interventions, and replacement interventions. The core results are direct. If the old equilibrium survives, it persists. This immediately limits the power of many taxes and subsidies. It also limits the effectiveness of reforms that merely add new options while leaving incumbent ones in place. By contrast, deletion and replacement interventions can induce transition because they undermine the strategic foundations of the old equilibrium. Replacement is often especially effective because it combines disruption with direction: it not only removes the old option but also helps coordinate behavior around a new one.

The economics is broad. In climate transition, carbon pricing may be desirable and still insufficient when incumbent technologies are embedded in complementary production systems. In digital markets, conduct remedies may leave platform architecture unchanged and therefore fail to alter the equilibrium. In financial reform, charging more for risk may be less effective than removing specific instruments or replacing them with safer contractual forms. In industrial modernization, subsidizing innovation may not move production networks if firms, suppliers, and workers remain coordinated on old standards. The same logic also appears in technological rivalry and business disruption. In the theory of disruptive innovation, established firms often improve existing products through sustaining innovations, while entrants change the technological trajectory itself. From the viewpoint developed here, sustaining innovation is close to an addition or price-based intervention, whereas disruptive innovation is naturally interpreted as a replacement intervention: the old technological path is not merely improved, but displaced by a new action space \citep{bower1995,christensen1997}.

Our contribution is threefold. First, we offer a compact framework that merges equilibrium selection with institutional persistence. Second, we provide a clear theoretical distinction between interventions that alter payoffs and interventions that alter feasible action spaces. Third, we use that distinction to clarify a practical policy lesson: under inertia, structurally targeted reforms may dominate softer incentive-based approaches.

The paper is also related to classic themes in institutional economics. \citet{north1990} emphasized that institutions persist because they stabilize expectations and organize repeated interaction. Our model translates that intuition into a strategic-form environment with equilibrium selection. More generally, the paper contributes to a growing recognition that many policy problems are transition problems. The relevant issue is not simply whether a better equilibrium exists, but how a society can move to it. For example, climate negotiations illustrate how institutional redesign may partially shift a strategic environment away from free-riding and toward coordinated participation; the Paris Agreement has been interpreted in this direction because it created a repeated pledge-and-review process and a focal framework for national commitments \citep{finus2020}.

The examples used in this paper are not intended as independent causal tests of the model. Rather, they serve as structured illustrations of the mechanism identified by the theory: whether an intervention leaves the inherited equilibrium intact, deletes an action that supports it, or replaces that action with a new focal alternative. This distinction helps organize examples from climate policy, disruptive innovation, military conflict, platform competition, financial reform, and technological rivalry without claiming that all such cases share the same institutional details.

The rest of the paper is organized as follows. Section 2 presents the formal framework, including the status-quo equilibrium, the post-intervention game, and the status-quo inertia selection rule. Section 3 defines the four intervention classes: price-only, addition, deletion, and replacement interventions. Section 4 states the main theoretical results. Section 5 provides stylized game-theoretic illustrations based on deletion and replacement. Section 6 presents a simple coordination example comparing price-only, addition, and replacement interventions within a common payoff environment. Sections 7 and 8 discuss illustrative applications from military and geopolitical cases, and from economic and technological rivalries, respectively. Section 9 examines boundary cases involving addition without deletion. Section 10 discusses broader implications and the relation to the literature. Section 11 concludes. Complete proofs are collected in the Appendix.

\section{Framework}

Consider a finite normal-form game
\[
G=\bigl(N,(A_i)_{i\in N},(u_i)_{i\in N}\bigr),
\]
where \(N=\{1,\ldots,n\}\) is the finite set of players, \(A_i\) is the finite action set of player \(i\), and
\[
A:=\prod_{i\in N}A_i
\]
is the set of action profiles. For each player \(i\in N\), the payoff function is
\[
u_i:A\to\mathbb{R}.
\]
An action profile is denoted by
\[
a=(a_i)_{i\in N}\in A.
\]
For each player \(i\in N\), let
\[
A_{-i}:=\prod_{j\in N\setminus\{i\}}A_j
\]
denote the set of action profiles of all players except \(i\). For \(a_i\in A_i\) and \(a_{-i}\in A_{-i}\), we write \((a_i,a_{-i})\in A\) for the corresponding full action profile.

A profile \(q\in A\) is a Nash equilibrium of \(G\) if, for every player \(i\in N\),
\[
u_i(q_i,q_{-i})\geq u_i(a_i,q_{-i})
\qquad
\text{for all } a_i\in A_i.
\]
The set of Nash equilibria of \(G\) is denoted by
\[
NE(G):=
\left\{
q\in A:
u_i(q_i,q_{-i})\geq u_i(a_i,q_{-i})
\text{ for all } i\in N
\text{ and all } a_i\in A_i
\right\}.
\]

\begin{definition}[Status-quo equilibrium]
A \emph{status-quo equilibrium} is a Nash equilibrium
\[
q^{-}\in NE(G)
\]
interpreted as the inherited or prevailing outcome before intervention.
\end{definition}

\noindent
\textbf{Economic example.}
In an energy market, \(q^{-}\) may correspond to an equilibrium in which producers, transport providers, equipment suppliers, and consumers all coordinate on fossil-based technologies. Even if a cleaner equilibrium exists, the incumbent one may remain the prevailing arrangement because it is supported by past investments, complementary infrastructure, and expectations.

\medskip

An intervention transforms the original game \(G\) into a post-intervention game
\[
G^I=\bigl(N,(A_i^I)_{i\in N},(u_i^I)_{i\in N}\bigr),
\]
with action space
\[
A^I:=\prod_{i\in N}A_i^I.
\]
The intervention may change payoffs, feasible actions, or both.

To capture persistence, let
\[
\sigma(\,\cdot\,;q^{-})
\]
be an equilibrium selection rule that assigns to each admissible post-intervention game \(H\) an equilibrium
\[
\sigma(H;q^{-})\in NE(H),
\]
whenever \(NE(H)\neq\varnothing\). The argument \(q^{-}\) records the inherited equilibrium of the pre-intervention game.

\begin{definition}[Status-quo inertia]
The equilibrium selection rule \(\sigma\) satisfies \emph{status-quo inertia} if, for every admissible post-intervention game \(H\),
\[
q^{-}\in NE(H)
\quad\Longrightarrow\quad
\sigma(H;q^{-})=q^{-}.
\]
\end{definition}

\noindent
\textbf{Economic example.}
A labor market may remain locked in a low-training equilibrium even after a training subsidy is introduced, provided firms still expect low skill investment and workers still expect weak demand for upgraded skills. The old equilibrium persists because it remains strategically self-confirming.

\medskip

The assumption is intentionally reduced-form. It does not claim that the status quo is always selected in every strategic environment. It says only that if the inherited equilibrium survives an intervention as a Nash equilibrium of the post-intervention game, then it retains selection priority. This captures, in a compact way, frictions such as focality, legal continuity, habit, switching costs, sunk investments, network effects, and coordination failure.


\section{Main Results}

We begin with the general persistence result.

\begin{proposition}[Persistence under status-quo inertia]
Let \(I\) be any intervention, and let \(G^I\) be the post-intervention game. If
\[
q^{-}\in NE(G^I),
\]
then
\[
\sigma(G^I;q^{-})=q^{-}.
\]
\end{proposition}

\noindent
\textbf{Economic example.}
Suppose a regulator imposes modest compliance penalties on a dominant digital platform, but the existing architecture still makes the old conduct pattern mutually optimal for users, advertisers, and the platform itself. Then the platform ecosystem may remain at the same equilibrium. The intervention affects incentives but not enough to dislodge the inherited strategic arrangement.

\medskip

Proposition 1 has an immediate implication for standard price-based reform.

\begin{proposition}[The limit of price-only intervention]
Suppose a price-only intervention produces a game \(G^P\) such that
\[
q^{-}\in NE(G^P).
\]
Then
\[
\sigma(G^P;q^{-})=q^{-}.
\]
\end{proposition}

\noindent
\textbf{Economic example.}
A carbon tax may improve welfare comparisons across technologies, but if the incumbent fossil-based production system remains self-enforcing because suppliers, maintenance firms, infrastructure providers, and customers are all coordinated around it, then the economy may stay at the old equilibrium.

\medskip

Addition interventions face the same basic limitation.

\begin{proposition}[Addition may leave the old equilibrium intact]
Suppose an addition intervention produces a game \(G^A\) such that
\[
q^{-}\in NE(G^A).
\]
Then
\[
\sigma(G^A;q^{-})=q^{-}.
\]
\end{proposition}

\noindent
\textbf{Economic example.}
A government may subsidize the introduction of a superior technology, but if incumbent firms expect their suppliers and customers to continue using the old system, adoption may remain limited. A better option exists, yet the economy does not move because the inherited equilibrium still survives.

\medskip

Deletion interventions are different because they can destroy the strategic support of the old equilibrium.

\begin{proposition}[Deletion can force transition]
Suppose a deletion intervention produces a game \(G^D\) such that
\[
q^{-}\notin NE(G^D).
\]
Assume further that \(G^D\) has a unique Nash equilibrium \(q^*\), and that \(q^*\) is efficient among the feasible post-intervention outcomes. Then
\[
\sigma(G^D;q^{-})=q^*.
\]
\end{proposition}

\noindent
\textbf{Economic example.}
If a regulator bans a class of destabilizing short-term contracts and the remaining institutional environment admits a unique safer market equilibrium, then the system can move there. The policy succeeds not because it slightly changes incentives, but because it removes the old strategic option entirely.

\medskip

Replacement can be even more effective because it can destroy the old equilibrium while also creating a focal alternative.

\begin{proposition}[Replacement can force directed transition]
Suppose a replacement intervention produces a game \(G^R\) such that
\[
q^{-}\notin NE(G^R).
\]
Assume further that \(G^R\) has a unique Nash equilibrium \(q^*\), and that \(q^*\) is efficient among the feasible post-intervention outcomes. Then
\[
\sigma(G^R;q^{-})=q^*.
\]
\end{proposition}

\noindent
\textbf{Economic example.}
In industrial policy, eliminating an obsolete standard may help, but replacing it with a common modern standard can be much stronger because it gives firms, suppliers, and workers a focal point for coordination. The reform does not simply break the old equilibrium; it helps organize the new one.

\medskip

The uniqueness assumptions in Propositions 4 and 5 are imposed for transparency. If the post-intervention game has multiple Nash equilibria, the conclusion requires an additional post-intervention selection criterion, such as efficiency, risk dominance, stochastic stability, or an explicitly specified refinement. The central point remains that transition requires the equilibrium protected by status-quo inertia to cease to survive.


\section{Types of Intervention}

We now define four intervention classes. The distinction among them concerns whether the intervention changes only payoffs, adds new actions, deletes old actions, or replaces old actions with new ones.

\begin{definition}[Price-only intervention]
A \emph{price-only intervention} changes payoffs while leaving all feasible action sets unchanged. Formally, the post-intervention game \(G^P\) satisfies
\[
A_i^P=A_i
\qquad
\text{for all } i\in N.
\]
There exist transfer functions
\[
T_i:A\to\mathbb{R}
\]
such that the post-intervention payoff functions are
\[
u_i^P(a)=u_i(a)+T_i(a),
\qquad
a\in A.
\]
\end{definition}

\noindent
\textbf{Economic example.}
Carbon taxes, investment subsidies, penalties for misconduct, and tradable permit prices are price-only interventions. They alter returns to existing actions but leave the action menu intact.

\medskip

\begin{definition}[Addition intervention]
An \emph{addition intervention} introduces a new feasible action without removing old actions. If \(a'\notin A_k\) is added for player \(k\), then the post-intervention action sets are
\[
A_k^A=A_k\cup\{a'\},
\qquad
A_i^A=A_i
\quad
\text{for all } i\neq k.
\]
The post-intervention action space is
\[
A^A:=\prod_{i\in N}A_i^A.
\]
The post-intervention payoff functions
\[
u_i^A:A^A\to\mathbb{R}
\]
coincide with the original payoff functions on the original action space \(A\), while payoffs involving the newly added action are specified by the intervention.
\end{definition}

\noindent
\textbf{Economic example.}
Introducing a green technology, creating a public digital payment rail, or offering a safer contractual default without eliminating incumbent arrangements are addition interventions.

\medskip

\begin{definition}[Deletion intervention]
A \emph{deletion intervention} removes an action from the feasible set of at least one player. If \(a\in A_k\) is deleted for player \(k\), then the post-intervention action sets are
\[
A_k^D=A_k\setminus\{a\},
\qquad
A_i^D=A_i
\quad
\text{for all } i\neq k.
\]
The post-intervention action space is
\[
A^D:=\prod_{i\in N}A_i^D.
\]
The post-intervention payoff functions
\[
u_i^D:A^D\to\mathbb{R}
\]
are inherited from the original payoff functions by restriction to \(A^D\).
\end{definition}

\noindent
\textbf{Economic example.}
Banning a highly polluting fuel, prohibiting a predatory contractual clause, or disallowing a destabilizing financial instrument are deletion interventions.

\medskip

\begin{definition}[Replacement intervention]
A \emph{replacement intervention} removes an old action and simultaneously introduces a new action for at least one player. If \(a\in A_k\) is replaced by \(a'\notin A_k\), then the post-intervention action sets are
\[
A_k^R=(A_k\setminus\{a\})\cup\{a'\},
\qquad
A_i^R=A_i
\quad
\text{for all } i\neq k.
\]
The post-intervention action space is
\[
A^R:=\prod_{i\in N}A_i^R.
\]
The post-intervention payoff functions
\[
u_i^R:A^R\to\mathbb{R}
\]
are inherited from the original game on profiles not involving the new action, while payoffs involving the new action are specified by the intervention.
\end{definition}

\noindent
\textbf{Economic example.}
Replacing a dirty production standard with a clean mandatory standard, substituting an opaque auction rule with a transparent one, or replacing a harmful financial contract with a regulated alternative are replacement interventions.

\medskip

\begin{remark}[Why replacement is often especially effective]
Pure deletion may destroy the status quo and still leave agents uncertain about which remaining equilibrium to coordinate on. Replacement often performs better because it combines removal of the old action with the creation of a salient new alternative. In environments with complementarities or network effects, this guidance can be decisive.
\end{remark}

\noindent
\textbf{Economic example.}
A ban on internal combustion buses may be disruptive if transit authorities and suppliers are unsure which technology to coordinate around. Replacing the old fleet standard with a unified electric standard is often more effective because it removes ambiguity as well as inertia.
\section{Stylized Game-Theoretic Illustrations}

The previous results are general. This section gives two stylized examples showing how deletion and replacement interventions operate in familiar game-theoretic environments. The examples are not intended to introduce new solution concepts. They are used only to clarify the mechanism by which an intervention can eliminate the inherited equilibrium or redirect coordination toward a new one.

\subsection{Deleting defection in a Prisoner's Dilemma}

Consider a two-player Prisoner's Dilemma in which each player chooses either cooperation \(C\) or defection \(D\). The payoff matrix is
\[
\begin{array}{c|cc}
 & C & D\\
\hline
C & 3,3 & 0,4\\
D & 4,0 & 1,1
\end{array}
\]
The unique Nash equilibrium is
\[
q^{-}=(D,D),
\]
even though \((C,C)\) is Pareto superior. In this case, the inefficient equilibrium is sustained because defection is individually optimal for each player, given the action of the other player.

Now suppose an institutional designer removes the action \(D\) from the feasible action set of each player. This is a deletion intervention. The post-intervention action sets become
\[
A_1^D=A_2^D=\{C\}.
\]
Hence the only feasible action profile is
\[
(C,C).
\]
The original status-quo equilibrium \((D,D)\) is no longer feasible and therefore cannot remain a Nash equilibrium of the post-intervention game. If \((C,C)\) is the unique Nash equilibrium of the reduced game, Proposition 4 implies that the selected equilibrium becomes
\[
\sigma(G^D;q^{-})=(C,C).
\]
This example shows the basic logic of deletion: the transition is not caused by a marginal change in payoffs, but by the removal of the action that sustained the inefficient status quo.

\subsection{Replacement in a Stag-Hunt-type coordination problem}

Consider a coordination problem with two actions: a safe incumbent action \(H\) and a superior coordinated action \(S\). The payoff matrix is
\[
\begin{array}{c|cc}
 & H & S\\
\hline
H & 2,2 & 1,0\\
S & 0,1 & 4,4
\end{array}
\]
There are two pure Nash equilibria:
\[
(H,H)
\qquad \text{and} \qquad
(S,S).
\]
The equilibrium \((S,S)\) is Pareto superior, but suppose the inherited status quo is
\[
q^{-}=(H,H).
\]
Under status-quo inertia, the selected equilibrium remains \((H,H)\) as long as it remains a Nash equilibrium of the post-intervention game.

A replacement intervention removes the incumbent action \(H\) and introduces a new coordinated action \(M\), interpreted as a modernized or institutionally supported alternative. The new action is designed to serve the same functional role as the old one, but under a different institutional or technological standard. Formally, for each player,
\[
A_i^R=(A_i\setminus\{H\})\cup\{M\}.
\]
The inherited equilibrium \((H,H)\) is no longer feasible. If the replacement is designed so that the post-intervention game has a unique Nash equilibrium \(q^*\), and \(q^*\) is efficient among the feasible post-intervention outcomes, then Proposition 5 yields
\[
\sigma(G^R;q^{-})=q^*.
\]
This example illustrates why replacement can be stronger than deletion alone. Deletion removes the old equilibrium, but replacement also gives players a new focal alternative around which expectations can coordinate.

\section{A Simple Coordination Example}

The preceding examples show the logic of deletion and replacement in familiar game-theoretic settings. We now give a minimal numerical coordination example that compares price-only, addition, and replacement interventions within the same payoff environment.

Consider a two-player coordination game in which each player chooses the old action \(O\) or the new action \(N\). Payoffs are
\[
\begin{array}{c|cc}
 & O & N\\
\hline
O & 3,3 & 0,1\\
N & 1,0 & 5,5
\end{array}
\]

There are two pure Nash equilibria:
\[
(O,O) \quad \text{and} \quad (N,N).
\]

Suppose
\[
q^{-}=(O,O)
\]
is the inherited status quo. Under status-quo inertia, the selected equilibrium is $(O,O)$ even though $(N,N)$ is Pareto superior.

Now compare three interventions.

\begin{enumerate}[label=(\roman*)]

\item \textbf{Price-only intervention.}
Suppose a small subsidy raises the payoff from choosing $N$. If the subsidy is not large enough to make $O$ unattractive when the other player chooses $O$, then $(O,O)$ remains a Nash equilibrium. By Proposition 2, it remains selected.

\item \textbf{Addition intervention.}
Suppose a new action $N'$ is introduced with attractive payoffs if both players adopt it. If players can still choose $O$ and $(O,O)$ remains an equilibrium, then Proposition 3 implies that inertia preserves the old outcome.\item \textbf{Replacement intervention.}
Suppose $O$ is removed and replaced by $N'$. Then $(O,O)$ is no longer feasible. If the resulting game has a unique Nash equilibrium \(q^*\), and \(q^*\) is efficient among the feasible post-intervention outcomes, the economy transitions to \(q^*\) by Proposition 5.

\end{enumerate}

Figure 1 summarizes the logic.

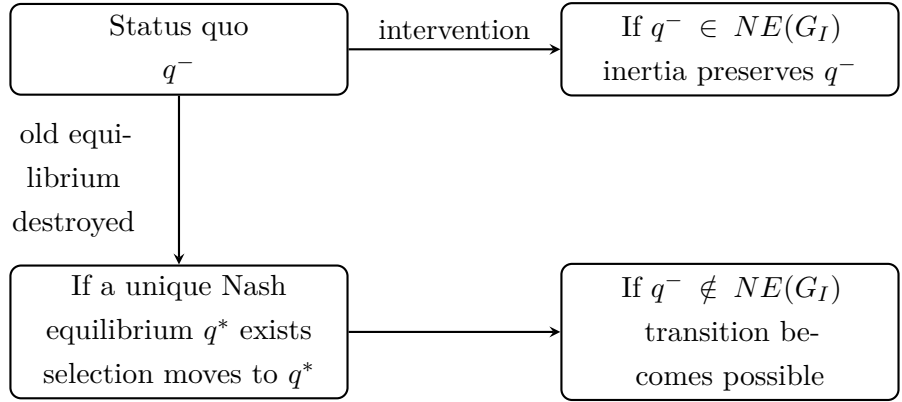
\begin{figure}[H]
\centering
\begin{tikzpicture}[
    node distance=2.2cm and 2.8cm,
    >=stealth,
    thick,
    box/.style={
        draw,
        rounded corners,
        align=center,
        text width=4.2cm,
        minimum height=1.2cm
    }
]

\node[box] (sq) {Status quo\\$q^{-}$};

\node[box, right=of sq] (survive)
{If $q^{-}\in NE(G_I)$\\inertia preserves $q^{-}$};

\node[box, below=of survive] (break)
{If $q^{-}\notin NE(G_I)$\\transition becomes possible};

\node[box, left=of break] (eff)
{If a unique Nash equilibrium $q^{*}$ exists\\selection moves to $q^{*}$};

\draw[->] (sq) -- node[above] {intervention} (survive);

\draw[->] (sq) -- node[left, text width=2.5cm, align=center] {old equilibrium destroyed} (eff);

\draw[->] (eff) -- (break);

\end{tikzpicture}

\caption{Status-quo inertia and equilibrium transition}
\end{figure}

Table 1 compares the intervention classes.

\begin{table}[H]
\centering

\caption{Interventions under status-quo inertia}

\small

\begin{tabular}{p{2.2cm}p{2.8cm}p{4cm}p{4.2cm}}

\toprule

Intervention &
What changes &
Typical failure mode &
Typical success condition \\

\midrule

Price-only &
Payoffs only &
The old equilibrium remains strategically viable &
Relative prices become strong enough to eliminate $q^{-}$ \\

Addition &
New actions are added, old ones remain &
Agents continue coordinating on incumbent actions &
The new option also undermines the old equilibrium \\

Deletion &
Old actions are removed &
Post-reform coordination may remain ambiguous &
The deleted action is essential for $q^{-}$ and the remaining game has a clear efficient equilibrium \\

Replacement &
Old actions are removed and new ones introduced &
Poor design can create a new coordination problem &
The reform both destroys $q^{-}$ and creates a focal efficient alternative \\

\bottomrule

\end{tabular}

\vspace{0.3em}

\begin{minipage}{0.95\textwidth}
\footnotesize
\textit{Notes:} The critical margin is whether the inherited equilibrium survives in the intervened game.
\end{minipage}

\end{table}

\section{Illustrative Applications: Military and Geopolitical Cases}

The preceding sections developed the formal mechanism and illustrated it through stylized games. This section provides selected applications showing how the same distinction between price-only, addition, deletion, and replacement interventions can organize real-world strategic transitions. The cases are not presented as independent causal tests of the theory. Rather, they are used as structured illustrations of the mechanism: whether an intervention leaves the inherited equilibrium intact, deletes an action that supports it, or replaces that action with a new focal alternative.

\subsection{The atomic bomb and the end of the Pacific War}

The atomic bombings of Hiroshima and Nagasaki in August 1945 provide an extreme example of a replacement-type intervention in a military setting. Before the intervention, the Japanese leadership still faced a strategic environment in which continued resistance, conditional surrender, and internal bargaining over terms remained feasible. The introduction of atomic warfare changed that environment by replacing the existing military-strategic calculus with a new one in which continued resistance carried a qualitatively different level of expected destruction. In the terminology of the present framework, the intervention did not merely alter a payoff at the margin; it changed the action space facing the Japanese leadership. The old strategic arrangement could no longer be sustained in the same form. Historical accounts emphasize that the bombings strengthened the position of the peace faction and contributed to the Emperor's decision to accept surrender \citep{asada1998}. This case therefore illustrates the replacement logic of Proposition 5: an inherited equilibrium can collapse when the strategic option supporting it is removed or transformed and a new outcome becomes focal.

\subsection{The Six-Day War and the deletion of air capability}

The opening phase of the Six-Day War in 1967 illustrates a deletion intervention. Israel's preemptive air operation targeted Egyptian air assets on the ground, thereby removing a key military action from the feasible set of the opposing coalition: the ability to conduct an effective conventional air campaign. In the language of the model, the relevant status quo depended on the availability of air cover and air retaliation. Once that action was removed, the previous strategic configuration could no longer function as before. Oren's account emphasizes the centrality of the June 5 airstrike in reshaping the battlefield within the first hours of the conflict \citep{oren2002}. The case therefore corresponds to Proposition 4: deleting an action that is essential to the inherited equilibrium can force transition to a different strategic state.

\subsection{The Yom Kippur War and the replacement of battlefield conditions}

The Yom Kippur War of 1973 provides a replacement-type illustration. After 1967, Israeli strategy relied heavily on the expectation of air superiority and armored mobility. Egyptian and Syrian forces, however, entered the 1973 war with Soviet-supplied surface-to-air missiles and anti-tank weapons. These weapons did not merely add another instrument to the battlefield; they changed the strategic environment in which air and armored forces operated. The previous action ``uncontested air and armor dominance'' was replaced by a battlefield in which those actions became costly and vulnerable. Rabinovich describes how these weapons contributed to unexpected Israeli losses in the early days of the war \citep{rabinovich2004}. In the present framework, this is a replacement intervention because the old operational equilibrium was displaced by a new action space with different constraints and expectations.

\subsection{The Gulf War and the deletion of centralized command}

The 1991 Gulf War provides another deletion illustration. Iraq's military position depended heavily on centralized command, control, communications, and intelligence infrastructure. The coalition air campaign targeted these nodes early in Operation Desert Storm. Once command-and-control capacity was degraded, the Iraqi military could no longer coordinate in the same way. The action ``centralized command'' was therefore removed or severely weakened as a feasible strategic action. Official accounts of the air campaign emphasize the importance of disabling Iraq's command-and-control network \citep{afhistory1991}. This case fits the deletion logic of Proposition 4: when the action supporting the inherited equilibrium is removed, the old strategic state becomes unsustainable.

\subsection{Stuxnet and cyber-physical replacement}

The Stuxnet operation illustrates a replacement intervention in a cyber-physical environment. The pre-existing strategic arrangement was based on the possibility of continuing uranium enrichment under conventional forms of monitoring, secrecy, and physical security. Stuxnet introduced a different action space by linking cyber intrusion to physical degradation of centrifuge equipment. In this sense, the intervention did not simply impose a cost on an existing action; it replaced the strategic environment of enrichment with one in which digital interference could produce physical consequences. Zetter's account describes Stuxnet as a digital weapon with effects beyond ordinary malware \citep{zetter2014}. In the language of Proposition 5, the inherited equilibrium was disrupted because an old action was no longer available under the same technological assumptions, while a new cyber-physical strategic environment became salient.

\subsection{Energy sanctions against Russia and replacement of supply dependence}

The energy sanctions and supply reconfiguration following Russia's 2022 invasion of Ukraine provide a geopolitical example of replacement. Before the invasion, parts of Europe were locked into an energy equilibrium based on long-standing dependence on Russian oil and natural gas. Sanctions and import restrictions did not merely raise the price of that dependence; they progressively removed the action ``continue normal energy trade with Russia'' from the feasible policy set of European governments. At the same time, alternative sources, including liquefied natural gas and accelerated renewable-energy investment, became focal substitutes. This case therefore combines deletion of the old strategic option with replacement by a new energy-security framework. Analyses of Europe's post-2022 gas strategy emphasize both the shift away from Russian pipeline dependence and the difficulty of replacing it with long-term LNG arrangements \citep{baker2023}. In the present framework, this illustrates a gradual replacement process rather than a simple price-only adjustment. The case illustrates the replacement logic of Proposition 5, while also showing that replacement may be gradual rather than instantaneous.


\section{Illustrative Applications: Economic and Technological Rivalries}

The same logic also appears in market competition and technological change. In these settings, status-quo inertia is often sustained by consumer habits, complementary assets, distribution systems, platform expectations, installed bases, and technological standards. The following cases illustrate how deletion and replacement interventions can reshape market equilibria, while also showing that such interventions are not necessarily welfare-improving in every case.

\subsection{Netflix and Blockbuster}

The rivalry between Netflix and Blockbuster illustrates replacement in a consumer market. Blockbuster's status quo was based on physical video rental stores, local inventory, and late fees. Netflix initially entered through DVD-by-mail, but the major equilibrium shift occurred when streaming became the dominant mode of access. Streaming did not merely add another rental channel; it replaced the physical-store logic with an on-demand digital distribution model. In the terminology of this paper, the action `rent a physical copy from a local store'' was progressively displaced by the action `stream content directly.'' Blockbuster's attempt to adapt while preserving much of the old architecture was closer to an addition intervention, whereas Netflix's model ultimately changed the action space of the market \citep{shih2014}. This illustrates Proposition 5: replacement can defeat inertia when it removes the strategic basis of the incumbent equilibrium and creates a new focal mode of coordination.

\subsection{The Pepsi Challenge}

The Pepsi Challenge illustrates a non-physical form of deletion. In the soft-drink market, Coca-Cola's status quo was supported not only by taste but also by brand loyalty, habit, and consumer identification. The blind taste-test format removed, at least temporarily, the action ``choose by brand recognition.'' Consumers were placed in a setting where brand identity was hidden and taste comparison became the salient action. This did not delete a physical product, but it deleted an informational or psychological support of the inherited equilibrium. In the framework developed here, the case shows that deletion can operate through the removal of cues, defaults, or informational structures that sustain status-quo behavior \citep{gladwell2005}.

\subsection{The iPhone and the displacement of incumbent smartphone designs}

The introduction of the iPhone provides a replacement-type illustration in technological rivalry. Before 2007, leading smartphone designs were organized around physical keyboards, carrier control, and relatively closed software environments. The iPhone replaced this architecture with a touchscreen interface, a software-centered device, and an application ecosystem. This was not merely an incremental improvement to the existing smartphone format. It changed the action space for consumers, developers, and manufacturers. The incumbent equilibrium based on keyboard-centered smartphones became increasingly unstable because the market began coordinating around a different technological standard \citep{naughton2017}. In the language of Proposition 5, the replacement succeeded because it both weakened the old equilibrium and supplied a new focal alternative.

\subsection{Samsung's TRX Malaysia campaign}

Samsung's campaign around The Exchange TRX in Malaysia provides a retail and marketing illustration of replacement. The relevant status quo was an Apple-centered customer journey around the opening of Apple's first official store in Malaysia. Apple opened its first retail store in Malaysia at The Exchange TRX in June 2024, while reports described Samsung's use of the adjacent MRT station as a separate branded advertising space around the Galaxy line \citep{edge2024}. The intervention did not merely add a competing advertisement inside the same controlled space; it replaced the neutral transit environment with a Samsung-branded entry point. This case illustrates how replacement may operate in spatial and attention-based competition: the strategic action space facing consumers was altered before they reached the incumbent's controlled retail environment.

\subsection{Tesla and legacy automakers}

Tesla's rise illustrates gradual replacement rather than instantaneous deletion. The incumbent automotive equilibrium was based on internal-combustion engines, dealership networks, periodic model-year improvements, and limited post-sale software change. Tesla introduced a different architecture: electric powertrains, direct sales, software-centered vehicle management, and over-the-air updates. This did not immediately eliminate the old equilibrium, but it progressively weakened it by creating a new focal standard for vehicle design and consumer expectations. In this sense, the Tesla case fits the replacement logic of the model with an important qualification: in many industrial transitions, replacement is dynamic and partial before it becomes dominant \citep{stringham2015}.

\subsection{Intel and AMD}

The Intel--AMD case illustrates that deletion is not inherently welfare-improving. The European Commission's Intel decision described conditional rebates and purchasing arrangements that restricted the practical ability of original equipment manufacturers to source CPUs from AMD, including arrangements under which HP was required to purchase no less than 95 percent of its business desktop CPU needs from Intel \citep{ecintel2009}. In the present terminology, this can be interpreted as an anti-competitive deletion of the action ``purchase from AMD'' from the effective choice set of some downstream firms. Unlike the regulatory examples discussed earlier, this deletion did not remove an inefficient or harmful action for social benefit; rather, it restricted competition. The case is therefore useful because it clarifies a boundary of the theory: deletion and replacement are powerful mechanisms for changing equilibria, but their welfare consequences depend on which actions are removed, who removes them, and what equilibrium replaces the status quo.


\section{Boundary Cases: Addition Without Deletion}

The previous applications mainly illustrate deletion and replacement. It is also important to examine addition-only interventions, because Proposition 3 implies that adding a new action does not by itself eliminate the inherited equilibrium. The following cases illustrate both the weakness of addition-only reform and a boundary condition under which addition can partially succeed.

\subsection{Blockbuster's mail service}

Blockbuster's mail-order DVD service illustrates an addition intervention that failed to overturn the incumbent equilibrium. By adding a mail-delivery option while preserving much of its store-based business model, Blockbuster expanded the action set available to consumers without fully removing the old rental architecture. The inherited equilibrium based on local stores, physical inventory, and late-fee-based routines remained strategically viable. In the terminology of Proposition 3, the new action did not eliminate the old equilibrium. The case therefore helps distinguish addition from replacement: Netflix's eventual streaming model changed the market's action space, whereas Blockbuster's mail service largely coexisted with the old one \citep{shih2014}.

\subsection{Windows Phone}

Windows Phone provides another example of addition without sufficient equilibrium displacement. Microsoft introduced a new smartphone platform into a market already coordinated around iOS and Android. The intervention added another feasible action for consumers, developers, and device manufacturers, but it did not remove the network effects supporting the incumbent platforms. Because application developers, users, and manufacturers continued to coordinate around the existing ecosystems, the old equilibrium remained intact. This case is consistent with Proposition 3: addition alone may fail when the inherited equilibrium is sustained by strong complementarities and expectations \citep{idc2017}.

\subsection{Zelle and institutionally anchored addition}

Zelle illustrates a more favourable boundary case for addition. Unlike many addition-only interventions, Zelle was embedded directly inside existing bank accounts and mobile banking applications. The new payment option did not remove older payment methods, such as cards, cash, checks, PayPal, or other peer-to-peer systems. However, it was institutionally anchored to an infrastructure that users already relied on. This reduced adoption frictions and gave the added action a built-in coordination base. The case therefore qualifies Proposition 3: addition is unlikely to displace an inherited equilibrium when it is merely placed beside old options, but it can become more effective when the new action is tied to an unavoidable or highly trusted institution \citep{zelle2024}.

\section{Discussion and Relation to the Literature}

The framework developed here is related to three major traditions in economics and game theory.

First, it belongs to the theory of equilibrium selection. \citet{schelling1960} stressed the role of salience and focal points, while \citet{harsanyi1988} proposed formal criteria for choosing among multiple equilibria. Our focus is narrower and more institutional: when an equilibrium is inherited from the past, its survival in the post-intervention game may itself determine selection. The examples above illustrate this mechanism in different strategic environments. The common feature is not that the cases share identical institutional details, but that each can be organized around the same question: does the intervention leave the inherited equilibrium intact, or does it remove the action that supports it?

Second, the paper is closely related to mechanism design and implementation theory. Following \citet{hurwicz1972}, economists came to understand institutions as systems that structure information and incentives. \citet{myerson1982} and \citet{maskin1999} further clarified how desired outcomes can be induced through carefully designed rules. The present argument complements that literature by emphasizing a practical transition problem. In environments with status-quo inertia, improving payoffs may be insufficient if the old equilibrium remains feasible. What matters is not only whether the target outcome is desirable, but whether the institutional design removes or replaces the strategic basis of the inherited outcome.

Third, the analysis connects naturally to institutional economics and path dependence. \citet{north1990} argued that institutions persist because they shape expectations and lower the uncertainty of repeated interaction. Our model gives that claim a strategic interpretation. The status quo survives not merely because it exists, but because it remains an equilibrium around which agents can coordinate. This perspective helps explain why many real-world reforms disappoint. In climate policy, a tax may be welfare-improving and still too weak to move the economy away from incumbent technologies. In digital markets, conduct rules may leave the platform's equilibrium logic intact. In finance, higher charges for risky behavior may be less effective than banning a contract form or replacing it with a safer regulated alternative. In industrial modernization, a subsidy for adoption may fail if the entire supply network remains coordinated on old standards.

The applications also clarify the distinction between illustration and proof. The military, geopolitical, and business cases are not presented as formal empirical tests or causal identification exercises. They are used to show how the proposed categories can discipline the interpretation of strategic transitions. In some cases, such as deletion of battlefield capabilities or restriction of contractual forms, the relevant mechanism is the removal of an action. In others, such as disruptive innovation, streaming media, smartphone redesign, or energy-supply diversification, the mechanism is closer to replacement: the old action space is not merely priced differently, but reorganized around a new focal alternative.

The boundary cases involving addition-only interventions are especially important. They show that the theory is not simply a claim that every new option succeeds. Blockbuster's mail service and Windows Phone illustrate the weakness of adding a new action while incumbent routines, expectations, and network effects remain intact. Zelle illustrates a more favourable boundary condition: addition can become more effective when the new action is embedded in an institution that users already rely on. This qualification matters because it prevents the framework from overstating the power of deletion and replacement or understating the possible role of institutionally anchored addition.

The practical implication is a decision rule for policy and strategy. If the inherited equilibrium is inefficient, the first question is whether it remains a Nash equilibrium after the proposed intervention. If it does, status-quo inertia predicts persistence. The second question is whether the intervention deletes an action essential to the old equilibrium or replaces it with a focal alternative. If it does not, the intervention may improve local incentives without inducing transition. If it does, transition becomes possible, although the welfare effect depends on what new equilibrium replaces the old one. The Intel--AMD example is a useful warning: deletion can be anti-competitive if it removes a socially valuable action. Thus, the framework identifies a mechanism of transition, not an automatic welfare theorem.

More broadly, the paper suggests that a central task of modern applied game theory is to study transition under persistence. Economists often know what the better equilibrium looks like. The harder problem is to understand how institutions can move the economy there. Under status-quo inertia, the answer often requires changing the feasible action space rather than merely changing marginal incentives.
\section{Conclusion}

This paper has studied interventions in games when equilibrium selection exhibits status-quo inertia. The main message is simple but consequential. If the inherited equilibrium remains a Nash equilibrium after intervention, it remains selected. This sharply limits the effectiveness of many price-based and soft-option reforms. By contrast, deletion and replacement interventions can induce transition because they alter the feasible action space and thereby remove the strategic basis for persistence.

The analysis developed this idea in a finite normal-form framework, distinguished four intervention classes, and derived persistence and transition results under status-quo inertia. The stylized game-theoretic examples showed how deletion can eliminate an inefficient defection equilibrium and how replacement can redirect coordination by removing an old action and introducing a new focal alternative. The simple coordination example then compared price-only, addition, and replacement interventions within a single payoff environment.

The illustrative applications extended the same logic to military conflict, geopolitical adjustment, market competition, technological rivalry, and financial or contractual settings. These examples are not offered as independent causal tests of the theory. Rather, they show how the distinction between changing payoffs and changing feasible actions can organize a range of strategic transitions. The boundary cases involving addition-only interventions further clarify the limits of the framework: adding a new option may fail when the old equilibrium remains viable, but addition can become more effective when it is institutionally anchored to an existing and widely used infrastructure.

The broader lesson is that policy design in strategic environments must take institutional entrenchment seriously. It is often not enough to improve incentives at the margin. When inefficient equilibria are sustained by coordination, expectations, routines, and inherited structures, the decisive intervention may be one that changes the game itself. At the same time, deletion and replacement are not automatically welfare-improving. Their consequences depend on which actions are removed, which alternatives are introduced, and whether the post-intervention equilibrium is socially desirable.

Future work could enrich the framework in several directions: dynamic adjustment, stochastic stability, endogenous political resistance to game-changing interventions, and empirical applications with explicit identification strategies. A further extension would be to study partial replacement, where the old equilibrium is weakened gradually rather than eliminated immediately. But the central point is already clear. In economies and institutions with inertia, successful reform is often structural before it is marginal..

\appendix

\section*{Appendix: Proofs}

\subsection*{Preliminary observation}

Let \(\sigma\) satisfy status-quo inertia. By definition, for any admissible post-intervention game \(H\),
\[
q^{-}\in NE(H)
\quad\Longrightarrow\quad
\sigma(H;q^{-})=q^{-}.
\]

\begin{lemma}
For any admissible post-intervention game \(H\), if
\[
q^{-}\in NE(H),
\]
then
\[
\sigma(H;q^{-})=q^{-}.
\]
\end{lemma}

\begin{proof}
This follows immediately from the definition of status-quo inertia.
\end{proof}

\subsection*{Proof of Proposition 1}

\begin{proof}
Let \(I\) be any intervention and let \(G^I\) denote the resulting post-intervention game. Assume
\[
q^{-}\in NE(G^I).
\]
By status-quo inertia, whenever the inherited equilibrium remains a Nash equilibrium in the post-intervention game, it is selected. Therefore,
\[
\sigma(G^I;q^{-})=q^{-}.
\]
\end{proof}

\subsection*{Proof of Proposition 2}

\begin{proof}
A price-only intervention changes payoff functions but leaves all action sets unchanged. Let \(G^P\) denote the resulting post-intervention game and suppose
\[
q^{-}\in NE(G^P).
\]
Then Proposition 1 applies directly, yielding
\[
\sigma(G^P;q^{-})=q^{-}.
\]
Hence a price-only reform fails to change the selected equilibrium whenever the inherited status quo remains a Nash equilibrium after the payoff change.
\end{proof}

\subsection*{Proof of Proposition 3}

\begin{proof}
Let \(G^A\) be the game generated by an addition intervention. By assumption,
\[
q^{-}\in NE(G^A).
\]
The addition of new actions does not alter selection under status-quo inertia so long as the inherited equilibrium remains a Nash equilibrium. Therefore, by Proposition 1,
\[
\sigma(G^A;q^{-})=q^{-}.
\]
\end{proof}

\subsection*{Proof of Proposition 4}

\begin{proof}
Let \(G^D\) denote the game after a deletion intervention. By assumption,
\[
q^{-}\notin NE(G^D).
\]
Assume further that \(G^D\) has a unique Nash equilibrium \(q^*\), and that \(q^*\) is efficient among the feasible post-intervention outcomes. Since the equilibrium selection rule satisfies
\[
\sigma(G^D;q^{-})\in NE(G^D),
\]
and since \(NE(G^D)=\{q^*\}\), it follows that
\[
\sigma(G^D;q^{-})=q^*.
\]
The efficiency statement is not needed for uniqueness of selection; it records the welfare interpretation of the selected post-intervention equilibrium.
\end{proof}

\subsection*{Proof of Proposition 5}

\begin{proof}
Let \(G^R\) denote the game after a replacement intervention. By assumption,
\[
q^{-}\notin NE(G^R).
\]
Assume further that \(G^R\) has a unique Nash equilibrium \(q^*\), and that \(q^*\) is efficient among the feasible post-intervention outcomes. Since the equilibrium selection rule satisfies
\[
\sigma(G^R;q^{-})\in NE(G^R),
\]
and since \(NE(G^R)=\{q^*\}\), it follows that
\[
\sigma(G^R;q^{-})=q^*.
\]
Again, efficiency gives the welfare interpretation of the transition, while uniqueness of the Nash equilibrium gives the formal selection result.
\end{proof}

\subsection*{On the role of uniqueness}

The uniqueness assumptions in Propositions 4 and 5 are imposed for transparency. If the post-intervention game has multiple Nash equilibria, the conclusion requires an additional post-intervention selection criterion, such as efficiency, risk dominance, stochastic stability, or an explicitly specified refinement. The central message remains the same: transition requires that the equilibrium protected by status-quo inertia cease to survive.

\section*{Acknowledgments}

The authors would like to thank colleagues and reviewers for their valuable comments and constructive feedback, which helped improve the clarity and presentation of the paper.
\section*{Funding}

The authors received no external funding for this research.
\section*{Conflict of Interest}

The authors declare that they have no conflict of interest.

\section*{AI Assistance Statement}

The authors used artificial intelligence tools solely for language improvement, grammar correction, and stylistic refinement of the manuscript. All mathematical definitions, theoretical concepts, formal results, proofs, and original scientific contributions were developed independently by the authors. The intellectual content and core ideas of the paper are entirely original and belong to the authors.

\bibliographystyle{unsrtnat}
\bibliography{references}

\end{document}